\documentclass[11pt,letterpaper]{article}
\usepackage{color,ifpdf,latexsym,graphicx,url}
\usepackage[margin=1in]{geometry}
\usepackage{amsmath,amssymb,amsthm}
\usepackage{mathtools}
\usepackage{enumerate}
\usepackage{afterpage}
\usepackage{microtype}
\graphicspath{{figures/}}

\title{Complexity of Simple Folding of Mixed Orthogonal Crease Patterns}
\author{%
  \tabcolsep=1em
  
  \begin{tabular}{ccc}
    Hugo Akitaya%
      \thanks{Miner School of Computer \& Information Sciences, University of Massachusetts
        Lowell, USA, \protect\url{hugo_akitaya@uml.edu}}
  & Josh Brunner%
      \thanks{MIT Computer Science and Artificial Intelligence Laboratory,
        32 Vassar St., Cambridge, MA 02139, USA,
        \protect\url{{brunnerj,edemaine,dylanhen,vluo,tockman}@mit.edu}}
  & Erik D. Demaine\footnotemark[2]
  \\
    Dylan Hendrickson\footnotemark[2]
  & Victor Luo\footnotemark[2]
  & Andy Tockman\footnotemark[2]
  \end{tabular}%
}
\date{}

\newif\ifabstract
\abstracttrue
\newif\iffull
\ifabstract \fullfalse \else \fulltrue \fi

\usepackage{hyperref}
\hypersetup{breaklinks,bookmarksnumbered,bookmarksopen,bookmarksopenlevel=2}
{\makeatletter \hypersetup{pdftitle={Simple Foldability Hardness}}}

{\makeatletter
 \gdef\xxxmark{%
   \expandafter\ifx\csname @mpargs\endcsname\relax 
     \expandafter\ifx\csname @captype\endcsname\relax 
       \marginpar{xxx}
     \else
       xxx 
     \fi
   \else
     xxx 
   \fi}
 \gdef\xxx{\@ifnextchar[\xxx@lab\xxx@nolab}
 \long\gdef\xxx@lab[#1]#2{\textbf{[\xxxmark #2 ---{\sc #1}]}}
 \long\gdef\xxx@nolab#1{\textbf{[\xxxmark #1]}}
}

{\makeatletter \gdef\fps@figure{!htbp}}

\let\realbfseries=\bfseries
\def\bfseries{\realbfseries\boldmath}

\newtheorem{theorem}{Theorem}[section]
\newtheorem{lemma}[theorem]{Lemma}
\newtheorem{corollary}[theorem]{Corollary}

\usepackage{multirow,colortbl,hhline}
\usepackage[table]{xcolor}

\let\epsilon=\varepsilon
\def\defn#1{\textbf{\textit{\boldmath #1}}}

\begin{document}
\maketitle

\begin{abstract}
  Continuing results from JCDCGGG 2016 and 2017, we solve several new cases of
  the \defn{simple foldability problem} --- deciding which crease patterns
  can be folded flat by a sequence of (some model of) simple folds.
  We give new efficient algorithms for \defn{mixed} crease patterns,
  where some creases are assigned mountain/valley while others are unassigned,
  for all 1D cases and for 2D rectangular paper with orthogonal one-layer simple folds.
  By contrast,
  we show strong NP-completeness for mixed orthogonal crease patterns
  on 2D rectangular paper with some-layers simple folds, complementing a
  previous result for all-layers simple folds.
  We also prove strong NP-completeness for finite simple folds
  (no matter the number of layers)
  of unassigned orthogonal crease patterns on arbitrary paper,
  complementing a previous result for assigned crease patterns,
  and contrasting with a previous positive result
  for infinite all-layers simple folds.
  In total, we obtain a characterization of polynomial vs.\ NP-hard for all
  cases --- finite/infinite one/some/all-layers simple folds of
  assigned/unassigned/mixed orthogonal crease patterns on
  1D/rectangular/arbitrary paper ---
  except the unsolved case of infinite all-layers simple folds of
  assigned orthogonal crease patterns on arbitrary paper.
\end{abstract}

\section{Introduction}

In the well-studied \defn{simple foldability problem}
\cite{map,simple,infinitealllayers,Demaine-O'Rourke-2007},
we are given a \defn{crease pattern}
consisting of line-segment \defn{creases}, possibly \defn{assigned}
mountain or valley, on a 2D region called the \defn{piece of paper},
and are asked whether all of the creases can be folded
via a sequence of ``simple folds''.
Each \defn{simple fold} folds some set of layers of the piece of paper
around a single line by $\pm 180^\circ$
(thus preserving flatness of the folding).
In the \defn{1D} variation, the piece of paper is a 1D line segment and
the creases are points.

Many different models for simple folds and special cases for the simple
foldability problem have been considered, depending on the following
characteristics:
\begin{itemize}
\item How many layers of paper a simple fold can move at once.
  The most powerful model, \defn{some-layers}, allows the top or bottom $k$
  layers to be folded for any $k$.  Two more restrictive models are
  \defn{one-layer}, which permits folding only a single layer at a time
  (modeling thick material); and \defn{all-layers}, which requires folding
  all layers simultaneously.
\item Whether simple folds can be along \defn{finite} line segments (chords
  of the crease pattern), or must be along \defn{infinite} lines (modeling
  a half-plane/large flipping tool).
\item Whether some creases are assigned as needing to be folded mountain
  or valley.  In the more common \defn{assigned} and \defn{unassigned} cases,
  all creases either have or lack an assignment, while in the \defn{mixed} case
  \cite{infinitealllayers}, each crease may be mountain, valley, or unassigned.
\item What shape the paper is allowed to have.  In 1D, the only option is a
  line segment. In 2D, the two cases we consider are rectangles and
  arbitrary polygons, though all NP-hardness results hold even for simple orthogonal polygons.
\end{itemize}

We consider exclusively \defn{orthogonal crease patterns}, which contain only vertical and horizontal creases. Prior work \cite{map,simple} has also studied versions of the problem including diagonal creases.

Table~\ref{summary} presents all known complexities of problems based on these parameters,
including both previously known results and new results from this paper
(in bold).
Note that, for orthogonal crease patterns, the finite vs.\ infinite simple fold distinction only plays a role
with arbitrary paper: no distinction can be made in the 1D case, and equivalence for rectangular paper is given by \cite[Theorem~8]{simple}.

\begin{table}
  \renewcommand\r[2]{\multirow{#1}{*}{#2}}
  \renewcommand\c[2]{\multicolumn{#1}{c|}{#2}}
  
  \definecolor{header}{rgb}{0.29,0,0.51}
  \definecolor{hard}{rgb}{1,0.85,0.85}
  \definecolor{open}{rgb}{0.95,0.95,0.5}
  \definecolor{easy}{rgb}{0.85,0.85,1}
  \newcommand\header{\cellcolor{header}\color{white}}
  \renewcommand\between[1]{\raisebox{.6ex}{#1}}
  \newcommand\betweenskip{-1ex}
  \newcommand\Left{\clap{$\Leftarrow$}}
  \newcommand\Right{\clap{$\Rightarrow$}}
  \newcommand\Up{\between{$\Uparrow$}}
  \newcommand\Down{\between{$\Downarrow$}}
  \newcommand\UpContain{\between{\rotatebox[origin=c]{-90}{$\subset$}}}
  \newcommand\DownContain{\between{\rotatebox[origin=c]{90}{$\subset$}}}
  \renewcommand\And{\clap{\&}}
  \newcommand\open{\cellcolor{open}OPEN}
  \newcommand\poly{\cellcolor{easy}poly}
  \newcommand\strong{\cellcolor{hard}NP-comp.}
  \newcommand\secref[1]{\rm\small~(\S\ref{sec:#1})}
  \newcommand\unsecref{\phantom{\small~(\S0)}}
  \newcommand\citeref[1]{\rm\tiny~\raisebox{0.65ex}{\cite{#1}}}
  \newcommand\leftjust[1]{#1}
  
  \arrayrulecolor{white}
  \arrayrulewidth=2pt
  \tabcolsep=4pt
  
  \newlength\layerswidth
  \settowidth\layerswidth{Layers}
  
  \centering
  \begin{tabular}{|c|c|c|c|c|c|c|c|c|}
    \c2{} & \header            & & \header                     & & \c3{\header Arbitrary Paper}
    \\ \hhline{|>{\arrayrulecolor{white}}-|-|>{\arrayrulecolor{header}}->{\arrayrulecolor{white}}|-|>{\arrayrulecolor{header}}->{\arrayrulecolor{white}}|-|-|-|-|}
    \c2{} & \r{-2}{\header 1D} & \r{-2}{\clap{$\subset$}} & \r{-2}{\header Rectangular} & \r{-2}{\clap{$\subset$}} & \header Infinite & \And & \header Finite
    \\ \hline
    \header
    & \header Assigned
    & \poly & \Left & \poly \citeref{map} &
    & \strong \citeref{simple} & & \bf\strong \secref{unassigned finite}
    \\
    \header & \UpContain & \Up & & \Up & & \Down & & \Down
    \\[\betweenskip]
    \header
    & \header Mixed
    & \bf\poly\secref{1D mixed some} & \Left & \bf\poly \secref{rect mixed one} &
    & \leftjust{\strong} & & \leftjust{\strong}
    \\
    \header & \DownContain & \Down & & \Down & & \Up & & \Up
    \\[\betweenskip]
    \header\multirow{-4.3}{\layerswidth}{\centering\header One \\ Layer}
    & \header Unassigned
    & \poly & \Left & \poly \citeref{map} &
    & \strong \citeref{simple} & & \leftjust{\bf\strong \secref{unassigned finite}}
    \\
    \c9{}
    \\[\betweenskip]
    \header
    & \header Assigned
    & \poly & \Left & \poly \citeref{map} &
    & \strong \citeref{simple} & & \bf\strong \secref{unassigned finite}
    \\
    \header & \UpContain & \Up & & & & \Down & & \Down
    \\[\betweenskip]
    \header
    & \header Mixed
    & \bf\poly \secref{1D mixed some} & & \bf\strong \secref{rect mixed some} & \Right
    & \leftjust{\strong} & \And & \leftjust{\strong}
    \\
    \header & \DownContain & \Down & & & & \Up & & \Up
    \\[\betweenskip]
    \header\multirow{-4.3}{\layerswidth}{\centering\header Some \\ Layers}
    & \header Unassigned
    & \poly & \Left & \poly \citeref{map} &
    & \strong \citeref{simple} & & \leftjust{\bf\strong \secref{unassigned finite}}
    \\
    \c9{}
    \\[\betweenskip]
    \header
    & \header Assigned
    & \poly & \Left & \poly \citeref{map} &
    & \open & & \bf\strong \secref{unassigned finite}
    \\
    \header & \UpContain & \Up & & & & & & \Down
    \\[\betweenskip]
    \header
    & \header Mixed
    & \bf\poly \secref{1D mixed all} & & \strong \citeref{infinitealllayers} & \Right
    & \strong & \And & \strong
    \\
    \header & \DownContain & \Down & & & & & & \Up
    \\[\betweenskip]
    \header\multirow{-4.3}{\layerswidth}{\centering\header All \\ Layers}
    & \header Unassigned
    & \poly & \Left & \poly \citeref{map} & \Left
    & \poly \citeref{infinitealllayers} & & \leftjust{\bf\strong \secref{unassigned finite}}
    \\
  \end{tabular}
  \caption{Summary of prior results (cited) and new results (bold, with section numbers) about simple folding orthogonal crease patterns. ``$\subset$'' denotes containment between classes of instances.  ``$\Rightarrow$'' denotes implications between results that are immediate from containment.
  ``\&'' represents that multiple cells are both part of the adjacent ``$\subset$'' or ``$\Rightarrow$'' relation.
  Hardness results for arbitrary paper hold even when restricted to
  simple orthogonal polygons.}
  \label{summary}
\end{table}

\subsection{Results}
\label{sec:results}

Our new results can be summarized as follows:

\begin{enumerate}
\item \label{secaa:1D mixed}
  \textbf{One- \& some-layers mixed 1D; one-layer mixed rectangular (\S\ref{sec:1D mixed one+some}).}
  We adapt arguments from \cite{map}, providing new characterizations of when a 1D assigned (\S\ref{sec:1D assigned characterization}) or mixed (\S\ref{sec:1D mixed some}) crease pattern is flat-foldable. For one-layer folds on rectangular paper (\S\ref{sec:rect mixed one}), crossing creases can never be folded, so the problem reduces to 1D paper.
\item \label{secaa:1D mixed all}
  \textbf{All-layers mixed 1D (\S\ref{sec:1D mixed all}).}
  We show that folding the superficially foldable (i.e., foldable ignoring assignments) crease nearest an end of the paper preserves foldability of the crease pattern, giving an efficient greedy strategy for deciding simple foldability.
\item \label{secaa:rect mixed some}
  \textbf{Some-layers mixed rectangular (\S\ref{sec:rect mixed some}).}
  We show that the NP-hardness reduction from \cite{infinitealllayers} for all-layers mixed rectangular also works for some-layers.
\item \label{secaa:unassigned finite}
  \textbf{Unassigned arbitrary finite (\S\ref{sec:unassigned finite}).}
  We modify the NP-hardness proof from \cite{simple} for assigned crease patterns to work for unassigned crease patterns, by adding a component that enforces the relevant aspects of the assignment.
  We also fix a bug in the reduction of \cite{simple}
  (also present in \cite{map})
  that allowed the folding to ``cheat'' instead of solving 3-Partition.
  The fixed previous proof and our adaptation apply to any number of layers.
\end{enumerate}

\label{sec:remaining-case}
Notably, the only case that remains unsolved for orthogonal crease patterns is infinite all-layers simple folds of assigned crease patterns on arbitrary paper.

\section{1D Paper, Mixed Assignment, One- and Some-Layers Folds}
\label{sec:1D mixed one+some}

In this section, we present new polynomial-time results for
1D paper in the one- and some-layers models.
In Section~\ref{sec:1D assigned characterization}, we give a new characterization of which 1D assigned crease patterns are foldable in the one- or some-layers models. In Section~\ref{sec:1D mixed some}, we extend our new characterization to a polynomial-time characterization for 1D mixed-assignment one- and some-layers crease patterns. Finally, in Section~\ref{sec:rect mixed one}, we show that the 1D mixed-assignment one-layer result generalizes to rectangular paper. (The same is not true for rectangular some-layers; we show in Section~\ref{sec:rect mixed some} that this version is strongly NP-complete.)

\subsection{Characterization for 1D Paper, Full Assignment}
\label{sec:1D assigned characterization}

We begin by introducing a new characterization of flat-foldable 1D assigned crease patterns. Such crease patterns have previously been characterized \cite{map,Demaine-O'Rourke-2007}, but our approach is easier to adapt to the mixed case because it can be described directly in terms of the crease pattern, without making a fold and recursively considering intermediate states.

Recall that a 1D crease pattern consists of a horizontal line segment,
called the \defn{paper},
together with a finite set of points on the paper,
called \defn{creases}.
Refer to Figure~\ref{fig:1d-def}(a--b).
The paper has two endpoints called \defn{ends};
together, we refer to creases and ends as \defn{vertices}.
For two vertices $P$ and $Q$ with $P$ to the left of $Q$,
let $[P,Q]$ denote the \defn{interval} of paper between them.
If $P$ and $Q$ are \emph{consecutive} vertices in the pattern,
we call the interval a \defn{segment} and denote it by $P Q$.
An interval $[P,Q]$ has up to two \defn{flaps}:
the segments just left of $P$ and just right of~$Q$.
The \defn{far endpoint} of a flap of $[P,Q]$
is the flap's endpoint other than $P$ or~$Q$.
An interval $[P,Q]$ is \defn{interior} if neither of its endpoints
are ends of the paper, in which case it has exactly two flaps.

\begin{figure}
	\centering
	\includegraphics[scale=2]{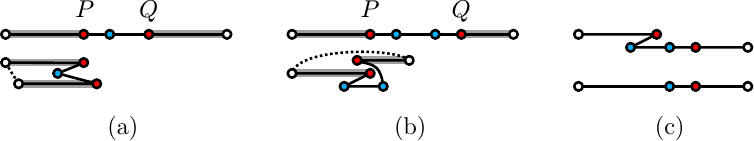}
	\caption{(a--b) Two examples of suspicious innocent intervals $[P,Q]$ with (a) an odd number of creases and (b) with an even number of creases. The flaps of $[P,Q]$ are highlighted in gray. Mountains (respectively, valleys) are drawn as red (respectively, blue) dots. The folded states drawn on the bottom are slightly deformed to convey the layer order. (c) A crimp executed on the two leftmost creases of the crease pattern in (b) and its resulting simplified crease pattern.}
	\label{fig:1d-def}
\end{figure}

Every point on the paper has a uniquely determined horizontal position
in the folded state: assuming the first segment remains stationary,
the horizontal position of each successive segment is determined by
alternating going left and right.
Let $f$ be the continuous function mapping the paper to its folded-state
geometry as above.
Call an interior interval $[P,Q]$ \defn{suspicious} if
the images (via~$f$) of its flaps' far endpoints
are both strictly outside the image of $[P,Q]$ in the folded state.
Call a suspicious interval $[P,Q]$ \defn{innocent}
if the numbers of mountain creases and valley creases in it
(including vertices $P$ and~$Q$) differ by at most~$\pm 1$.
Again refer to Figure~\ref{fig:1d-def}(a--b).
Our new characterization of flat foldability is simply:
\begin{center}
  A 1D assigned crease pattern is flat foldable $\iff$ Every suspicious intervals is innocent.
\end{center}

We prove flat foldability using two folding operations
from~\cite{map} that are parameterized by a segment $P Q$.
If $P Q$ is an interior segment that is nonstrictly shorter than both its flaps,
and $P$ and $Q$ are assigned opposite orientations
(one mountain and one valley),
then we call $P Q$ \defn{crimpable} and
define a \defn{crimp} to fold $P$ and then $Q$
(with their assigned orientations);
see Figure~\ref{fig:1d-def}(c).
If $P$ (respectively~$Q$) is an end of the paper,
and the segment $P Q$ is nonstrictly shorter than its one flap,
then we call $P Q$ \defn{end-foldable}
and define an \defn{end fold} to fold $Q$ (respectively~$P$).
Both crimps and end folds can be performed as one-layer simple folds,
and produce a folded state where unfolded creases
do not overlap any other unfolded point,
and thus the unfolded creases remain single-layer.
As a result, after applying these operations,
we can reduce the problem to folding a smaller crease pattern obtained by ``gluing'' the overlapping segments, which we call the \defn{result}
of the operation; see the bottom of Figure~\ref{fig:1d-def}(c).
Any sequence of simple folds in the smaller crease pattern
(in the one- or some-layers model) can be applied to the folded state,
and thus extended to the original crease pattern (in the same model,
because the crimped or end-folded creases remain single-layer).

\begin{theorem} \label{thm:1dchar}
For an assigned 1D crease pattern, the following are equivalent:
\begin{enumerate}
	\item The crease pattern is flat-foldable by a sequence of crimps and end folds.
	\item The crease pattern is flat-foldable by a sequence of one-layer simple folds.
	\item The crease pattern is flat-foldable by a sequence of some-layers simple folds.
	\item The crease pattern is flat-foldable.
	\item Every suspicious interval in the crease pattern is innocent.
\end{enumerate}
\end{theorem}
\begin{proof}
  We prove a cycle of implications:
  $1 \implies 2 \implies 3 \implies 4 \implies 5 \implies 1$.

\begin{description}
\item[$1\implies 2$:]
As described above, each end fold or crimp can be performed as a sequence
of one or two (respectively) one-layer simple folds,
and these operations preserve that every unfolded crease is single-layer.

\item[$2\implies 3$:]
By containment of models.

\item[$3\implies 4$:]
Simple folding produces flat foldings.

\item[$4\implies 5$:]
Consider a flat folding of the paper, and restrict attention to any suspicious interval and its flaps (removing the rest of the paper).
Add a line segment connecting the two far endpoints (the dotted curves in Figure~\ref{fig:1d-def}(a--b)).
Because the interval is suspicious, this line segment does not intersect the flat folding.

The flat folding plus the added line segment forms a loop,
with one added crease when the interval has an odd number of creases
(Figure~\ref{fig:1d-def}(a)) or
two added creases when the interval has an even number of creases
(Figure~\ref{fig:1d-def}(b)).
These added creases are convex relative to the loop, because they are at the leftmost or rightmost positions of the loop. A non-intersecting closed loop has net turn angle $360^\circ$, so the loop must have two more convex creases than reflex creases.

In the even case, we have added two convex creases and now have net two convex creases, so there must have previously been the same number of convex and reflex creases. In the odd case, we have added one convex crease, so there must have previously been one more convex than reflex. Therefore the interval is innocent.
(This argument is essentially a generalization of
\cite[Theorems 2 and 4]{Hull-2003-counting}.)

\item[$5\implies 1$:]

We prove some lemmas relating innocent suspicious intervals to crimps and end folds:

\begin{lemma}\label{lem:innocent -> crimp}
If every suspicious interval is innocent, then
some segment must be crimpable or end-foldable.
\end{lemma}
\begin{proof}
  Let $P Q$ be a shortest segment in the crease pattern.
  Consider the maximal sequence of equal-length segments
  containing $P Q$.
  If one end is an end of the paper, then the incident segment is end-foldable.
  If neither end is an end of the paper, then
  this interior interval of equal-length segments is suspicious.
  The interval is thus innocent by assumption,
  and contains at least two creases ($P$ and $Q$),
  so it must contain at least one mountain and at least one valley
  and in particular must contain adjacent mountain and valley creases.
  The segment between these two creases is crimpable
  because it has the same minimum length as $P Q$
  so the two neighboring flaps are at least as long.
\end{proof}

\begin{lemma}\label{lem:crimp preserves innocent}
Performing a crimp preserves that every suspicious interval is innocent.
\end{lemma}
\begin{proof}
Suppose we perform a crimp by folding adjacent creases $A$ and~$B$,
and assume that every suspicious interval is innocent before the crimp.
The segment $A B$ must be (nonstrictly) shorter than the two neighboring
segments, and the result of the crimp merges these three segments of paper
into one new segment.
Consider some suspicious interval $[P,Q]$ after the crimp,
and let its far endpoints be $X$ and $Y$, adjacent to $P$ and $Q$ respectively.
By assumption, the images of $X$ and $Y$ in the folded state lie
strictly outside the image of $[P,Q]$.
Consider the new segment's relation to $[P,Q]$,
as shown in Figure~\ref{fig:l2-3-cases}:
\begin{figure}
	\centering
	\includegraphics[scale=2]{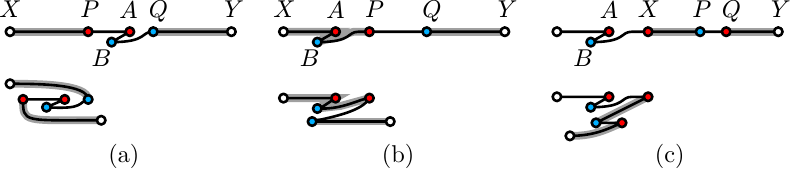}
	\caption{Cases in the proof of Lemma~\ref{lem:crimp preserves innocent}.}
	\label{fig:l2-3-cases}
\end{figure}

\begin{enumerate}[(a)]
  \item If $[P,Q]$ contains the new segment (as in Figure~\ref{fig:l2-3-cases}(a)), then crimping did not change the image of $[P,Q]$ or its flaps in the folded state. Thus $[P,Q]$ was suspicious before the crimp, so was innocent by assumption. The crimp does not change the net number of creases in $[P,Q]$ (because the crimp consumes one mountain and one valley), so $[P,Q]$ is still innocent.
  \item If the new segment is a flap of $[P,Q]$, assume by symmetry that it is the flap adjacent to $P$, and the vertices $X$, $A$, $B$, and $P$ occur consecutively in that order before crimping (as in Figure~\ref{fig:l2-3-cases}(b)).  Consider how $B$ relates to $[P,Q]$ in the folded state:
  \begin{itemize}
    \item Suppose that, in the folded state, the image of $B$ is strictly outside the image of $[P,Q]$. Then $[P,Q]$ is suspicious before folding the crimp, because its far endpoints are $B$ and $Y$.
    \item Now suppose that the image of $B$ is inside the image of $[P,Q]$. By crimpability, $A B$ is no longer than $B P$, so the image of the interval $[A,P]$ is entirely contained in the image of $[P,Q]$. Because $[P,Q]$ is suspicious after crimping, $X$ and $Y$ land outside the image of $[P,Q]$, so they both also lie outside the image of $[A,Q]$. Thus $[A,Q]$ is suspicious before folding the crimp.
  \end{itemize}
  In either case, one of $[P,Q]$ and $[A,Q]$ was suspicious and thus innocent before the crimp. The number of creases of each orientation (mountain or valley) in $[P,Q]$ is the same before and after the crimp, and $[A,Q]$ has one additional crease of each orientation. Thus, in either case, $[P,Q]$ is innocent after the crimp.
  \item Otherwise, the crimp lies outside of $[P,Q]$ and its flaps (as in Figure~\ref{fig:l2-3-cases}(c)), and thus does not affect this interval, so $[P,Q]$ remains innocent. \qedhere
\end{enumerate}
\end{proof}

\begin{lemma}\label{lem:end fold preserves innocent}
Performing an end fold preserves that every suspicious interval is innocent.
\end{lemma}
\begin{proof}
By symmetry consider the case where the end fold is the leftmost crease.
The result of this end fold just deletes the leftmost segment
of the crease pattern.
This deletion destroys any intervals containing the leftmost crease,
and does not affect the innocence of any intervals
not containing the leftmost crease.
\end{proof}

To prove $5\implies 1$, we induct on the number of creases.
Because every suspicious interval is innocent,
Lemma~\ref{lem:innocent -> crimp} gives us a crimpable or end-foldable segment.
Thus we find a crimp or end-fold to perform first.
By Lemmas~\ref{lem:crimp preserves innocent} and
\ref{lem:end fold preserves innocent}, this operation
maintains that every suspicious interval is innocent
in the resulting crease pattern, which also has fewer creases.
By induction, we obtain a sequence of crimps and end folds to fold the result.
Prepending the first operation folds the crease pattern.
\end{description}

The cycle of implications completes the proof of Theorem~\ref{thm:1dchar}.
\end{proof}

\subsection{1D Paper, Mixed Assignment, One-Layer and Some-Layers Folds}
\label{sec:1D mixed some}

We now consider mixed 1D crease patterns, and give a polynomial-time algorithm to determine flat-foldability.

An \defn{assignment} for a mixed 1D crease pattern specifies whether each unassigned crease should be mountain or valley. We call such an assignment \defn{valid} if the resulting assigned crease pattern is flat-foldable --- or equivalently, by Theorem~\ref{thm:1dchar}, if every suspicious interval is innocent.

\begin{theorem}\label{thm:1d mixed foldable}
  Given a 1D mixed crease pattern, we can in polynomial time determine whether it has a valid assignment, and find one if it exists.
\end{theorem}

\begin{proof}
  Whether an interval is suspicious does not depend on the assignment, only the positions of creases. Finding a valid assignment for the crease pattern is thus a constraint satisfaction problem, where each suspicious interval is a clause requiring that it is innocent. 

  At a high level, the algorithm performs these steps:
  \begin{enumerate}
  	\item \label{stp:1} Find all suspicious intervals, by checking all (quadratically many) pairs of creases. 
  	\item\label{stp:2} For each suspicious interval $I$ from smallest to largest, assign all or all-but-one still-unassigned creases in $I$, in a way that guarantees the innocence of $I$, and preserves the existence of a valid assignment.
      This step may fail with a report that no valid assignment exists.
  	\item \label{stp:3} Assign the remaining creases arbitrarily, and return the resulting valid assignment.
  \end{enumerate}

  To demonstrate a method of assigning suspicious intervals that preserves the existence of a valid assignment (Lemma~\ref{lem:invariant}), we will use the following lemma:

  \begin{lemma}
  \label{suspicious-intersection}
  The intersection of two suspicious intervals is suspicious.
  \end{lemma}
  \begin{proof}
  The image of the intersection of two intervals in the folded state is contained in the intersection of the images of the intervals in the folded state. 
  The right flap of the intersection is the same as the right flap of the left interval, so its far endpoint lies outside the image of the left interval, and thus also outside the image of the intersection.
  Similarly, the far endpoint of the intersection's left flap also lies outside the image of the intersection, so the intersection is suspicious.
  \end{proof}

  Define a \defn{minimal} suspicious interval to be one that contains no smaller suspicious intervals that have not yet been fully assigned.
  Our algorithm will repeatedly take a minimal suspicious interval, and assign all or all but one of the unassigned creases in it to guarantee its innocence.

  For each minimal suspicious interval $I$, we assign its creases as follows:

  \begin{itemize}
    \item If $I$ has an even number $2k$ of creases, we check whether its unassigned creases can be assigned to make exactly $k$ mountain and $k$ valley creases in $I$, i.e., whether the previously assigned creases have at most $k$ of each sign. If such as assignment is possible, use any such assignment; otherwise, we terminate the algorithm and report that there is no valid assignment.
    \item If $I$ has an odd number $2k+1$ of creases, we check whether it is possible to exclude one unassigned crease $P$ and find an assignment on the rest, to make $k$ mountain and $k$ valley creases among $I\setminus\{P\}$. 
    If this is possible, use any such assignment on the other creases and leave $P$ unassigned.
    Otherwise, if there is a full assignment that makes the interval innocent, use any such assignment. (This case happens only when there are $k+1$ previously assigned creases of one sign, so there is only one such assignment.)
    If this is also not possible, terminate the algorithm and report that there is no valid assignment.
  \end{itemize}

  We repeat this process until every suspicious interval has been considered.
  At this point, all suspicious intervals are guaranteed to be innocent,
  so we can assign all remaining creases arbitrarily,
  and return this as a valid assignment by Theorem~\ref{thm:1dchar}.
  Conversely, if the algorithm ever reports that there is no valid assignment,
  then there is no valid assignment by the following lemma:

  \begin{lemma}
    \label{lem:invariant}
    Assigning the creases for a single minimal suspicious interval $I$
    as described above preserves the existence of a valid assignment.
  \end{lemma}
  
  \begin{proof}
    Let $\mathcal{A}$ be the assignment before the algorithm processes
    suspicious interval~$I$, and let $\mathcal{A}_I$ be the assignment
    after the algorithm processes~$I$. 
    By assumption, $\mathcal{A}$ can be extended into a valid assignment $\mathcal{A}^+$.

    We construct an assignment $\mathcal{A}_I^+$ that extends $\mathcal{A}_I$,
    agrees with $\mathcal{A}^+$ everywhere except within interval~$I$,
    and assigns the same numbers of mountains (respectively valleys) to~$I$
    as $\mathcal{A}^+$ does.
    Starting from $\mathcal{A}^+$, we flip the assignment of each crease
    in $\mathcal{A}^+$ that does not agree with $\mathcal{A}_I$
    (which must be within~$I$).
    By definition of innocent, if $I$ has an even number $2 k$ of creases,
    then both $\mathcal{A}_I^+$ and $\mathcal{A}^+$ assign $k$ mountains
    and $k$ valleys to~$I$.
    Otherwise, $I$ has an odd number $2 k+1$ of creases.
    If the algorithm was forced to make a full assignment, it was because
    $\mathcal{A}$ already had $k+1$ mountains or valleys, so
    $\mathcal{A}_I^+$ and $\mathcal{A}^+$ agree on mountain/valley counts.
  	Otherwise, $\mathcal{A}_I$ left one crease $P$ in $I$ unassigned.
  	If the mountain/valley crease counts of $\mathcal{A}_I^+$ and
    $\mathcal{A}^+$ differ,
    then flip the assignment of $P$ in $\mathcal{A}_I^+$,
    making the counts the same while remaining consistent with $\mathcal{A}_I$.

    Now we show that $\mathcal{A}_I^+$ is a valid assignment.
    By Theorem~\ref{thm:1dchar}, it suffices to show that every
    suspicious interval $J$ is innocent in $\mathcal{A}_I^+$
    (given the same for~$\mathcal{A}^+$).
    If $J$ contains $I$, then $J$ is innocent:
    the mountain/valley counts in $\mathcal{A}_I^+$
    differ from $\mathcal{A}^+$ only within $I$, and
    within $I$ the counts are the same by construction.
    If $J$ does not contain $I$, then their intersection $I \cap J$
    is a smaller suspicious interval by Lemma~\ref{suspicious-intersection}.
    Because $I$ was a minimal suspicious interval, $I \cap J$ was already
    assigned in $\mathcal{A}$.
    Therefore $\mathcal{A}_I^+$ and $\mathcal{A}^+$ agree on $I \cap J$
    and thus~$J$, so $J$ remains innocent.
  \end{proof}

  Finally we discuss the running time.
  Step~\ref{stp:1} takes $O(n^3)$ time: for a given pair of creases, we can compute the image of the interval in the folded state and check whether the interval is suspicious in $O(n)$ time.
  Step~\ref{stp:2} can be done with one linear scan on the interval to count the number of assigned creases of each type and another linear scan to assign the unassigned creases, once for each of $O(n^2)$ suspicious intervals, for a total of $O(n^3)$ time.
  Step~\ref{stp:3} takes $O(n)$ time.
  This concludes the proof of Theorem~\ref{thm:1d mixed foldable}.
\end{proof}

By Theorem~\ref{thm:1dchar}, the algorithm in
Theorem~\ref{thm:1d mixed foldable} to find flat-foldable assignments
is sufficient to determine whether a mixed crease
pattern is flat-foldable in the one-layer and some-layers models:

\begin{corollary}\label{cor:1d mixed one}
  Given a mixed 1D crease pattern, we can determine in polynomial time whether it can be folded flat using one-layer simple folds.
\end{corollary}

\begin{corollary}\label{cor:1d mixed some}
  Given a mixed 1D crease pattern, we can determine in polynomial time whether it can be folded flat using some-layers simple folds.
\end{corollary}

\subsection{Rectangular Paper, Mixed Assignment, One-Layer Folds}
\label{sec:rect mixed one}

As another consequence of Theorem~\ref{thm:1d mixed foldable}, we can extend Corollary~\ref{cor:1d mixed one} by replacing the 1D paper with rectangular paper.

\begin{corollary}\label{cor:rect mixed one}
  Given a mixed crease pattern on rectangular paper, we can determine in polynomial time whether it can be folded flat using one-layer simple folds.
\end{corollary}
\begin{proof}
  An orthogonal crease pattern on rectangular paper that contains both
  horizontal and vertical creases is never flat-foldable in the one-layer model,
  because we cannot fold two intersecting creases with one-layer simple folds.
  On the other hand, if there are only creases in one direction, then the
  crease pattern is equivalent to a 1D mixed creased pattern, so we can use the
  algorithm from Theorem~\ref{thm:1d mixed foldable}.
\end{proof}

\section{1D Paper, Mixed Assignment, All-Layers Folds}
\label{sec:1D mixed all}

In this section, we describe a polynomial-time greedy algorithm for mixed 1D crease patterns in the all-layers model.


In all-layers simple folding, layers can never be separated once they come into
contact, so we can view the paper as getting glued together where it overlaps.
Thus a \defn{valid fold} in 1D
must fold each valley (respectively mountain) crease
onto a mountain (respectively valley) or unassigned crease,
or onto a point not interior to the paper:
folding a crease onto a non-vertex point of paper (neither crease nor end),
or a crease with an incompatible (equal) assignment,
would prevent future all-layers folding of the crease.
(Equal assignments are incompatible
because one crease gets flipped upside-down during a simple fold.)
Equivalently, we can discard the shorter portion of the paper
on either side of the crease,
reducing the paper to a shorter line segment with fewer creases \cite{map}.
See Figure~\ref{1d greedy fail}(a, c).

\begin{figure}
  \centering
  \includegraphics[scale=2]{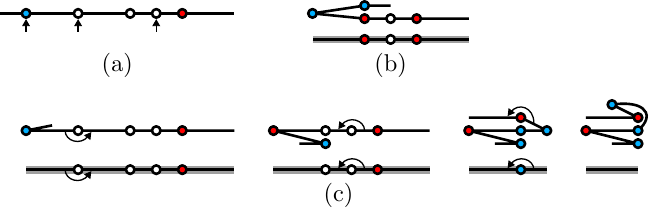}
  \caption{An example instance of 1D mixed all-layers simple folding,
    demonstrating the failure of the simple greedy algorithm.
    (a) Crease pattern, where
    blue, red, and gray dots denote valleys, mountains, and unassigned
    creases, respectively.
    Arrows indicate all valid and all plausible creases.
    (b) If we first fold the leftmost unassigned crease,
    then the result cannot be folded, because it would require
    folding together creases with incompatible assignments.
    (c) Folding the valley crease on the left first
    allows for a valid folding.
    Below the actual foldings we draw the reduced 1D equivalents.}
  \label{1d greedy fail}
\end{figure}

In the fully assigned and fully unassigned cases, these ideas
lead to the greedy algorithm presented in \cite{map,infinitealllayers}:
at each step, perform any valid fold.
For any valid fold, the new paper must be a subset of the original paper,
with no assignments changed.
Thus, all valid folds preserve foldability,
and the greedy algorithm will always succeed if there is a solution. 

Unfortunately, this greedy algorithm can fail in the mixed case.
Folding an assigned crease onto an unassigned crease effectively causes
it to become assigned.
In fact, a foldable crease pattern can become unfoldable
after making a valid fold.
For example, in Figure~\ref{1d greedy fail},
folding the leftmost unassigned crease as in Figure~\ref{1d greedy fail}(b)
results in an unfoldable crease pattern.

To fix this problem with mixed crease patterns,
we give a more refined greedy algorithm.
Call a crease \defn{plausible} if it would be a valid (first)
fold if every crease were unassigned.
That is, a crease $P$ is plausible if the locations of creases
(ignoring assignment) are symmetric around $P$
up to (and not including) the nearest end of the paper.
For example, Figures~\ref{1d greedy fail}(a) and~\ref{fig:interval}
mark the plausible creases with arrows.
(These examples of plausible creases are also all valid folds,
though they need not be, if symmetric creases have equal assignments. 
The center crease in the bottom crease pattern of Figure~\ref{1d greedy fail}(b) is an example of a plausible crease that is not valid.)

\begin{lemma} \label{lem:plausible}
  Let $P$ be a plausible crease with minimum distance to an end of the paper.
  If there is a sequence of valid all-layers folds that folds the crease
  pattern, then there is such a sequence that folds $P$ first.
\end{lemma}

\begin{proof}
Assume by symmetry that $P$ is on the left half of the paper
(including the midpoint), and thus is the leftmost plausible crease.
Refer to Figure~\ref{fig:interval} for an example.
Let $d$ be the distance from $P$ to the left end of the paper,
and consider the open interval $I_P$ of length $2 d$ centered at $P$.
Because $P$ is plausible, the creases in $I_P$ are symmetric around $P$.

\begin{figure}
  \centering
  \includegraphics[scale=3]{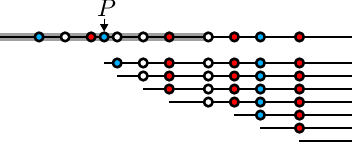}
  \caption{An example 1D mixed crease pattern,
    which can be folded using all-layers simple folds
    (as shown by the reduced folding sequence).
    Blue, red, and gray dots denote valleys, mountains, and unassigned
    creases, respectively.
    The original crease pattern has only one plausible crease~$P$,
    whose corresponding interval $I_P$ is shaded.}
  \label{fig:interval}
\end{figure}

Suppose we are given a valid all-layers folding sequence
$F=f_1, f_2, \dots, f_n$,
where each $f_i$ folds one crease in the current reduced folding,
which corresponds to a set of creases in the original crease pattern
that have been identified by previous folds.

Consider the first fold $f_i \in F$ that folds a crease in~$I_P$,
and let $Q$ be the (only) crease in $I_P$ folded by~$f_i$.
Before fold $f_i$, $I_P$ cannot have changed, except for creases
gaining assignments because of creases outside $I_P$ being folded onto them.
Thus, if we restrict attention to the paper in $I_P$,
$Q$ must be foldable in the original crease pattern as well.
In particular, $Q$ must be plausible when 
we restrict the crease pattern to~$I_P$.

Because the locations of creases in $I_P$ are symmetric around~$Q$,
the reflection $Q'$ of $Q$ across $P$ is also plausible within $I_P$.
If $Q \neq P$, then one of $Q$ and $Q'$ lies to the left of $P$.
Because the left end of $I_P$ was originally the left end of the paper,
this crease must be plausible even when considering the whole paper.
But we assumed $P$ is the leftmost plausible crease, a contradiction.
Therefore $Q=P$.
In other words, the first fold $f_i$ in $I_P$ is at~$P$.

The folding $F$ induces a mountain/valley assignment for all creases,
based on whether they actually folded mountain or valley. 
Because $f_i$ is the first fold in $I_P$, this full assignment restricted
to $I_P$ must be antisymmetric (map mountains to valleys and vice versa)
when reflected across~$P$.
Therefore $P$ is a valid fold in this fully assigned crease pattern.
Because valid all-layers folds can be made in any order
for assigned crease patterns,
there is a valid sequence of all-layers folds that folds $P$ first.
\end{proof}

Note that the plausible crease $P$ in Lemma~\ref{lem:plausible}
(one with minimal distance to an end of the paper)
still might not be valid.
But in this case, we can deduce from Lemma~\ref{lem:plausible}
that the crease pattern cannot be folded using all-layers simple folds.
Indeed, we can turn this into a refined greedy algorithm:
repeatedly find, and fold if possible, the plausible crease nearest an end.

\begin{theorem}
There is a polynomial-time algorithm to determine whether a 1D mixed crease pattern can be folded using all-layers simple folds.
\end{theorem}

\begin{proof}
  In linear time, we can check whether a crease is plausible.
  Thus, in quadratic time, we can find all plausible creases.
  If there are no plausible creases, then the crease pattern cannot be folded,
  because valid folds (considering assignment) are certainly plausible
  (ignoring assignment).
  Otherwise, we take the (leftmost or rightmost) plausible crease $P$
  nearest an end of the paper.
  We check whether $P$ is a valid all-layers folds.
  By Lemma~\ref{lem:plausible}, if $P$ is not valid, then we can report
  that the crease pattern cannot be folded; and if $P$ is valid,
  then we can safely fold $P$ first and preserve foldability
  (if we had it to begin with).
  Then we repeat the above procedure on the reduced crease pattern
  with fewer creases, so the number of rounds is at most linear.
  Therefore, in cubic time,%
  \footnote{This algorithm can be optimized to run in quadratic time
    by searching for plausible creases in order from the outside in,
    so that when we find one $k$ creases from the left or right end,
    the fold reduces the problem by $k$ creases.
    We leave open whether the algorithm can be further optimized
    using additional ideas from \cite{map,infinitealllayers}.}
  this algorithm will either produce a valid folding sequence,
  or it will fail to find valid crease at some step,
  in which case the original crease pattern was not foldable.
\end{proof}

\section{Rectangular Paper, Mixed Assignment, Some-Layers Folds}
\label{sec:rect mixed some}

In this section, we prove that it is NP-complete to determine whether a
mixed orthogonal crease pattern on rectangular paper has a simple folding
in the some-layers model.
This result contrasts the polynomial-time algorithm of
Corollary~\ref{cor:rect mixed one} for the same problem in the one-layer model.

We adapt the reduction from \cite[Section~5]{infinitealllayers}
which shows NP-hardness of the same problem in the all-layers model.
In fact, we show that the instances described in their reduction
never permit a folding other than the intended one,
even when relaxing to the some-layers model,
and thus their reduction also applies to our case.
Because the some-layers model is more permissive than the all-layers model,
it suffices to show is that no unintended foldings are possible.

We begin with a lemma which provides a situation where the some-layers model is actually no more permissive than the all-layers model.

\begin{lemma}
  \label{some layers all layers rect}
  Whenever a simple fold is made in an orthogonal crease pattern with
  rectangular piece of paper, if it is not in the same direction
  (horizontal or vertical) as the previous fold, it must be an all-layers fold.
\end{lemma}
\begin{proof}
  We prove this claim by induction.
  As a base case, the first fold clearly is an all-layers fold
  because there is only one layer to fold.

  Now suppose the claim is true for all of the folds made so far.
  Assume by symmetry that the current fold is a vertical fold
  and the previous fold is a horizontal fold.

  Consider any two layers of the paper at the fold line.
  We will show that, if a fold $f$ goes through one of these layers,
  it must go through the other layer.
  Because the layers are on top of each other,
  there must be some previous fold $f'$ which folded them on top of each other.
  If $f'$ is horizontal,
  then $f'$ must intersect the current vertical fold $f$ at a point~$x$,
  because the folds are infinite.
  Around point~$x$, if only one of these layers gets folded by~$f$,
  it will tear the paper along that previous crease.
  On the other hand, if $f'$ is vertical, then by assumption
  there must have been a horizontal fold $f''$ since $f'$.
  By the induction hypothesis, $f''$ must have folded both of these layers.
  Thus, the same argument that the paper would tear
  when folding $f$ still holds.
\end{proof}

\begin{theorem}
  It is NP-complete to determine whether a mixed orthogonal crease pattern
  on rectangular paper be folded in the some-layers model.
\end{theorem}

\begin{proof}
We begin with a brief overview of Akitaya et al.'s reduction;
refer to Figure~\ref{old-mixed-figure} and
see \cite[Section~5]{infinitealllayers} for the full proof.
The reduction is from 3SAT. We have a long strip of rectangular paper, divided into a square grid of creases. The paper is divided into $n$ sections each corresponding to a variable. 
The section of variable $x_i$ is between lines $f_{i-1}$ and $f_{i}'$ in Figure~\ref{old-mixed-figure}.
The lower part of a variable section contains one valley crease for each occurrence of the variable in a clause (creases adjacent to a yellow circle in Figure~\ref{old-mixed-figure}), and the upper part contain one mountain crease for each occurrence of the variable in a clause (creases adjacent to a green or red circle depending on whether the corresponding literal is positive or negative).
Additionally, there is a topmost section that checks the satisfiability of the clauses, and a set of assigned creases (on vertical lines $v_i$ and $v_i'$ for $i\in\{0,\ldots,n\}$) enforcing the variable order.

\begin{figure}[t]
  \centering
  \includegraphics[scale=1.4]{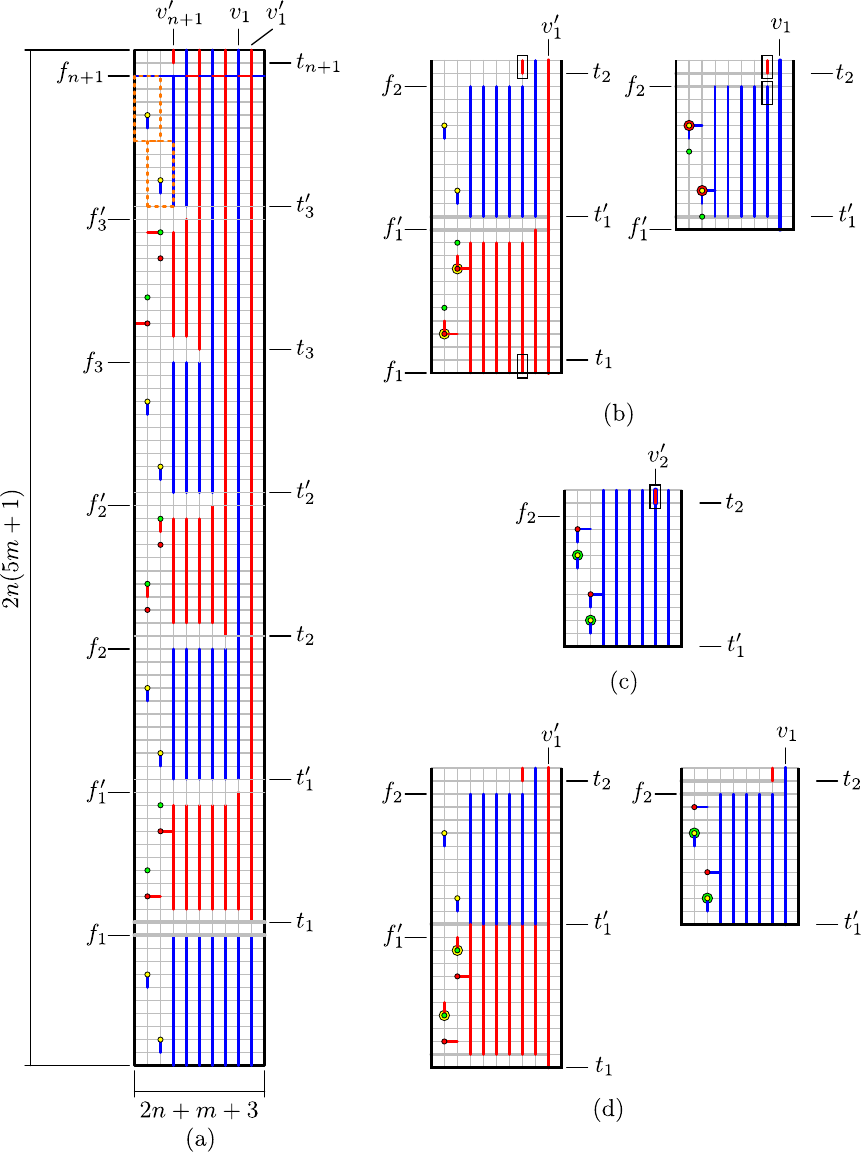}
  \caption{Figure 5 from \cite{infinitealllayers} [used with permission]:
    reduction from 3SAT to all-layers simple foldability of
    mixed-assignment orthogonal crease pattern on rectangular paper. The instance corresponds to the boolean formula $ (x_1 \vee \overline{x_2} \vee x_3) \vee (x_1 \vee \overline{x_2} \vee \overline{x_3})$.}
  \label{old-mixed-figure}
\end{figure}

At each step of the reduction, there are two adjacent unassigned creases corresponding to setting the variable true or false ($f_i$ and $t_i$), either of which can be folded first, which will fold this variable's assigned creases into the next section. Which of these creases is folded first determines whether the variable is assigned \textsc{true} or \textsc{false}. At the end, any clause with three False variables will have four valley folds coming from a single vertex, which is unfoldable in any model. After folding one of the horizontal creases in a section, the only legal fold is the rightmost remaining vertical crease. This forces the variables to be folded in order and prevents any other foldings.

During the variable assigning stage of the reduction, at each point, the only legal folds are the two horizontal folds corresponding to that variable. The only fold that is allowed in the some-layers model, and not in the all-layers model, is folding \emph{both} of the creases for a variable in a crimp. However, the entire length of this crimp will contain two oppositely assigned vertical creases on top of each other, immediately making it unsolvable.

After each pair of horizontal folds for a variables, there is a vertical fold.
By Lemma~\ref{some layers all layers rect},
we know that each of these vertical folds is an all-layers fold.
The following variable's first fold is also an all-layers fold,
so later variables do not have any more flexibility than earlier ones.

At the end of the variable folding, the clauses must also be folded. Because of the alternating vertical and horizontal folds from the variable folding, all of the folding of the clauses will have to fold all of the layers of the paper from the variable folding by Lemma~\ref{some layers all layers rect}. Every \textsc{no} instance in the reduction cannot be folded because of a degree-4 vertex with all valley assignments. Because we cannot separate the layers with these valley folds from each other, there is no way to fold this even in the some-layers model because the number of mountains and valleys must differ by exactly two around any vertex.

Thus, every \textsc{no} instance in their reduction is not foldable in the some-layers model, so the reduction still shows NP-hardness even in the some-layers model.
\end{proof}

\section{Arbitrary Paper, No/Full Assignment, Finite Folds}
\label{sec:(un)assigned finite}

In this section, we prove that it is NP-complete to determine whether an unassigned or assigned orthogonal crease pattern on arbitrary 2D paper has a simple folding.
This result applies to all three models of simple folds: one-layer, some-layers, and all-layers.
Our proof is based on modifying the construction in \cite{simple}, which reduces 3-Partition to the assigned problem.
However, this reduction has a bug, so we start with a description
of the bug (Section~\ref{sec:assigned bug})
and a fix (Section~\ref{sec:assigned fix}),
before adapting the proof to the unassigned case
(Section~\ref{sec:unassigned finite}).

\subsection{Bug in Previous Assigned Reduction}
\label{sec:assigned bug}

We first give a brief overview of Akitaya et al.'s reduction from 3-Partition
\cite{simple}.
Recall the \defn{3-Partition} problem: given $n=3 m$ numbers
$a_1, \dots, a_{3 m}$, divide them into $m$ triples of equal sum, namely,
$t = (\sum_{i=1}^n a_i)/3$.

\begin{figure}
  \centering
  \scalebox{0.8}{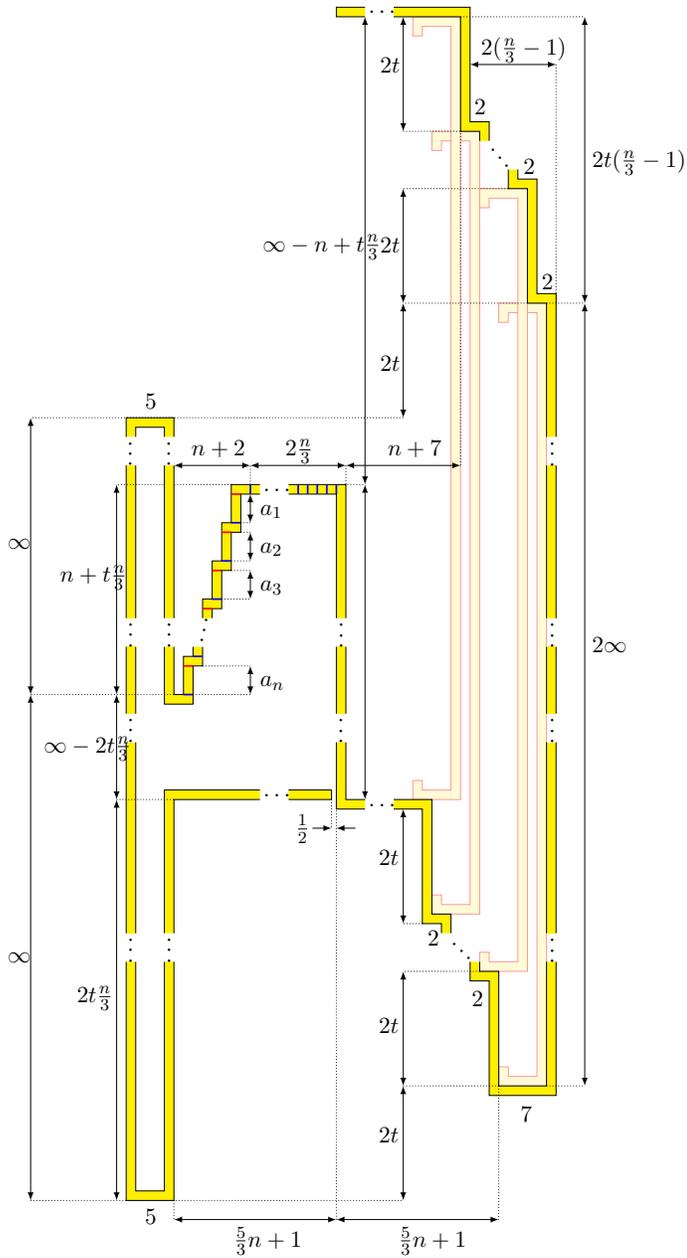}
  \caption{Figure 2 from \cite{simple} [used with permission].
    Buggy reduction from 3-Partition to assigned orthogonal crease pattern
    simple foldability with arbitrary paper.}
  \label{old-orthogonal-finite-figure}
\end{figure}


The construction consists of six parts, shown in Figures~\ref{old-orthogonal-finite-figure} and~\ref{old-orthogonal-finite-steps-figure}, which are called the Bar, the Staircase, the Wrapper, the Column, the Cage, and the Arm.
The Bar is a very long rectangle of paper with no creases.
The Staircase has a series of vertical segments with lengths $a_i$, for each $a_i$ in the 3-Partition instance, which are separated by horizontal creases, alternating mountain and valley.
The Wrapper is a horizontal rectangle with $2 m$ valley creases,
and crease spacing equals the width of the adjacent vertical Column;
thus, it wraps around that Column repeatedly.
The Cage has no creases and contains two ``staircases'', where each segment is of height $2t$. 
The Arm has no creases and is a long horizontal bar below the Wrapper attached to the Bar.

\begin{figure}[t]
  \centering
  \includegraphics[width=\linewidth]{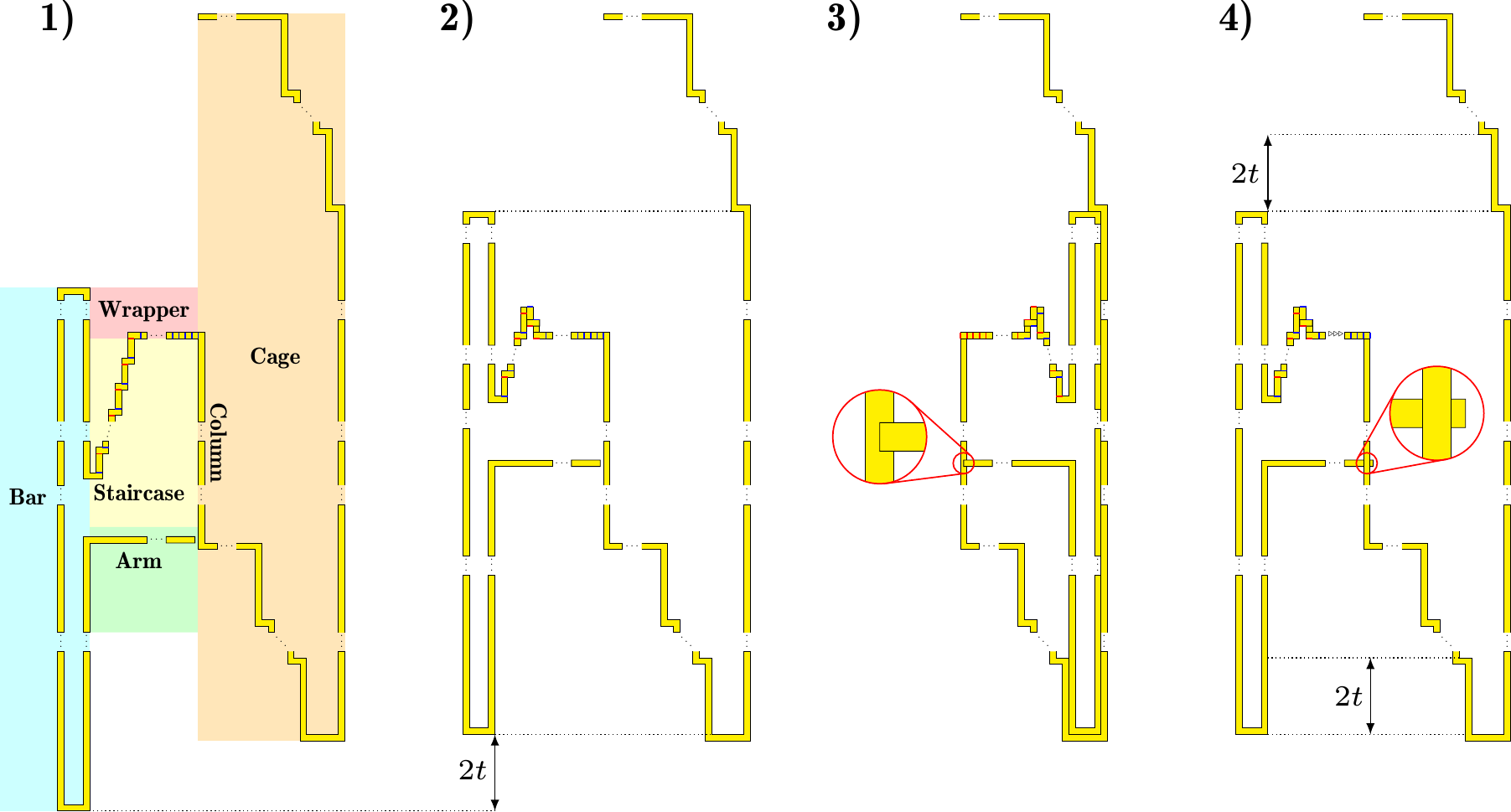}
  \caption{Figure 3 from \cite{simple} [used with permission].
  Process to implement a 3-Partition solution:
  1)~crimp variables to change height of bar by $2t$;
  2)~fold along the rightmost Wrapper crease around the Column;
  3)~fit the Bar through the Cage folding the Bar to the left along the next
  Wrapper crease; and
  4)~repeat until $n/3$ triples adding to $2t$ have been checked.}
  \label{old-orthogonal-finite-steps-figure}
\end{figure}

In order to fold this construction in the intended fashion, the Bar must pass through the Cage without colliding with it each time two of the creases on the Wrapper is folded. In order to achieve this, the vertical position of the Bar must align with the staircases in the Cage. Because each successive step in the Cage is $2t$ higher than the previous, they require folding $t$ of total length of segments in the Staircase each time. This can only be achieved if the 3-Partition instance has a solution. Figure~\ref{old-orthogonal-finite-steps-figure} shows one step of folding the Bar through the Cage after crimping $2t$ total height of the Staircase segments. 

Unfortunately, there is an unintended folding that avoids vertically aligning
the Bar with each successive step of the Staircase.
Refer to Figure~\ref{old-orthogonal-bug}.
Instead of folding two creases of the Wrapper in succession,
we fold just one of those creases.
Suppose that we have folded an odd number of Staircase creases,
and that the Bar is vertically higher than the corresponding step of the Cage.
Then the Bar will overlap the top part of the Cage,
preventing a second Wrapper crease from being folded.
But it is possible to instead fold a horizontal valley crease of the Staircase,
provided it is below the top of the Cage.
This fold might not cause a collision if it pulls the overlapping part of the Bar upwards, away from the Cage. Note that the rotation axis in Figure~\ref{old-orthogonal-bug}(c) is below the overlap between the Bar and the Cage.
Assume that the Bar continues to overlap the top part of the Cage.
Such overlap can only happen if the originally top part of the bar (which was below the rotation axis of the last fold) is now beneath the cage. See Figure~\ref{old-orthogonal-bug}(d).
This allows a second Wrapper crease (which is a mountain in this view)
to be folded.
In this way, we can avoid solving 3-Partition in at least some cases.

\begin{figure}
  \centering
  \includegraphics[width=\linewidth]{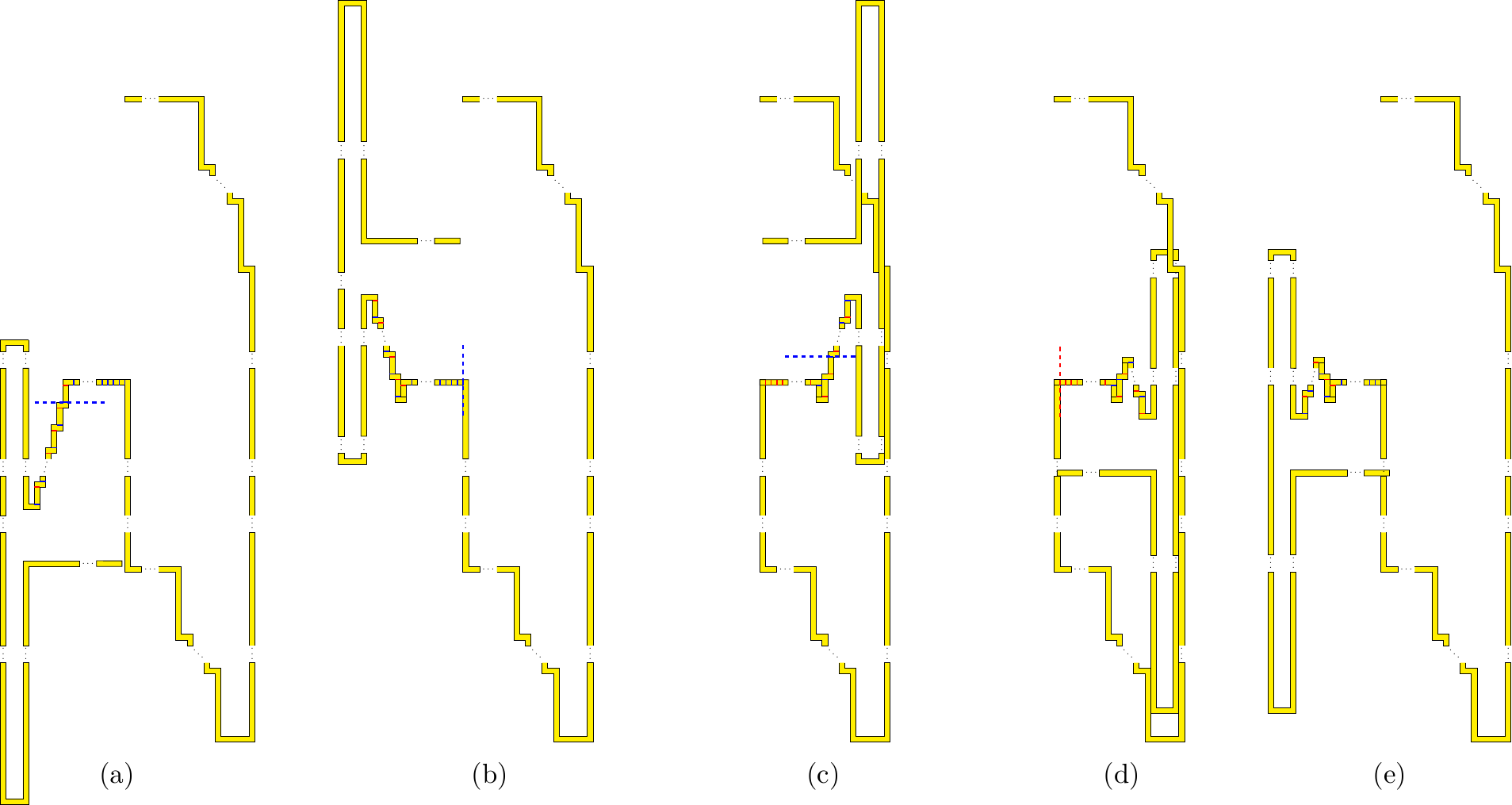}
  \caption{Unintended folding which breaks the reduction from
    Figures~\ref{old-orthogonal-finite-figure}
    and~\ref{old-orthogonal-finite-steps-figure}. The axis of the folds are highlighted with dashed blue and red lines.
    First, we fold some of the Staircase creases (a).
    Second, we fold a first Wrapper crease (b).
    Third, we fold a Staircase crease to bring the Bar below the Cage (c).
    Fourth, we fold a second Wrapper crease (d).}
  \label{old-orthogonal-bug}
\end{figure}

\subsection{Assigned Reduction}
\label{sec:assigned fix}

Next we show how to correct the proof of this theorem from \cite{simple}:

\begin{theorem} \label{thm:assigned fix}
It is NP-complete to determine whether an assigned orthogonal crease pattern on arbitrary (or orthogonal) paper can be folded in each of the one-layer, some-layers, and all-layers models.
\end{theorem}

\begin{proof}
  Our modified reduction from 3-Partition adds a second Arm
  \emph{above} the Wrapper,
  as illustrated in Figures~\ref{new-orthogonal-finite-figure}
  and~\ref{new-orthogonal-finite-steps-figure}.
  This prevents the bad behavior in Figure~\ref{old-orthogonal-bug}:
  in the middle diagram, the new Arm 2 would point left from the
  bottom of the Bar, overlapping 
  the Column, which prevents the valley fold in the Staircase.
  Indeed, any fold through a Staircase crease would cause the Bar to collide with the Cage unless they are exactly aligned.
  Note that the added arm does not prevent the originally intended folding sequence,
  as illustrated in Figure~\ref{new-orthogonal-finite-steps-figure}.
\end{proof}

\begin{figure}
  \centering
  \includegraphics{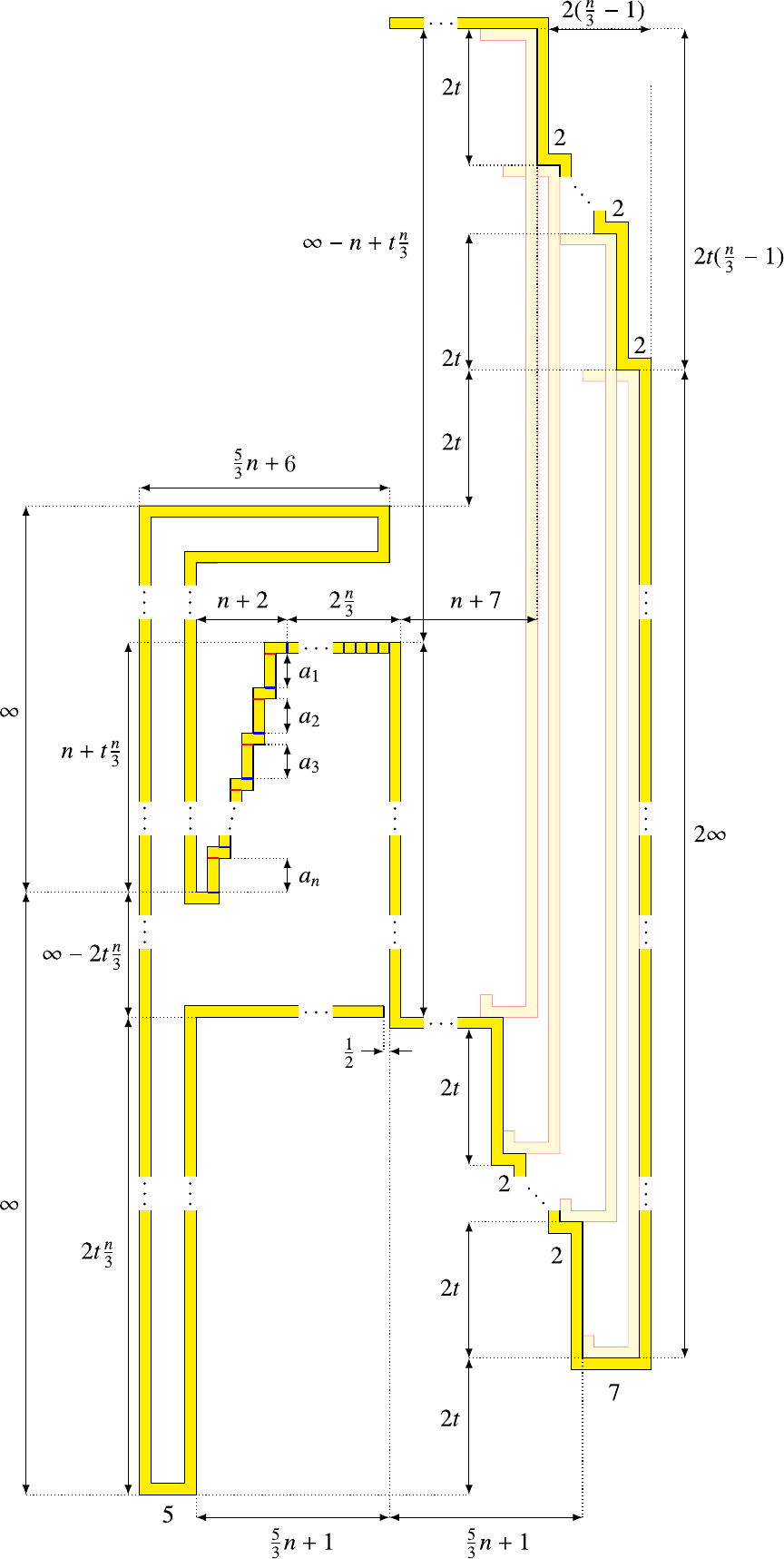}
  \caption{Reduction from 3-Partition to assigned orthogonal crease pattern
    simple foldability with arbitrary paper,
    based on Figure~\ref{old-orthogonal-finite-figure}.}
  \label{new-orthogonal-finite-figure}
\end{figure}

\begin{figure}[t]
  \centering
  \includegraphics[width=\linewidth]{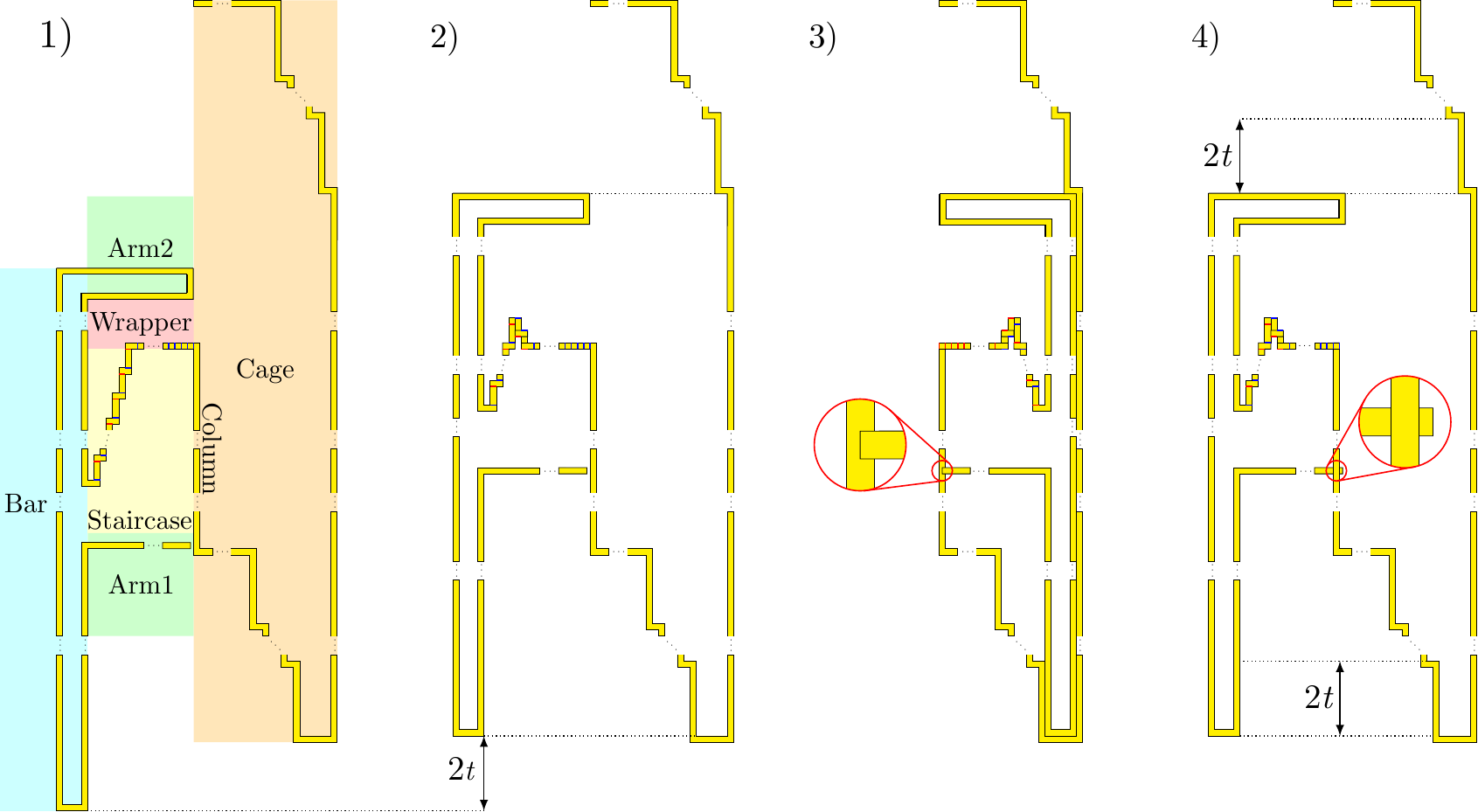}
  \caption{Intended folding sequence for the new reduction in
    Figure~\ref{new-orthogonal-finite-figure} with the added (and highlighted)
    Arm~2, based on Figure~\ref{old-orthogonal-finite-steps-figure}.
    1)~crimp variables to change height of bar by $2t$;
    2)~fold along the rightmost Wrapper crease around the Column;
    3)~fit the Bar through the Cage folding the Bar to the left along the next
    Wrapper crease; and
    4)~repeat until $n/3$ triples adding to $2t$ have been checked.}
  \label{new-orthogonal-finite-steps-figure}
\end{figure}

\subsection{Unassigned Reduction}
\label{sec:unassigned finite}

Finally we modify the proof to the unassigned case:

\begin{theorem}
It is NP-complete to determine whether an unassigned orthogonal crease pattern on arbitrary (or orthogonal) paper can be folded in each of the one-layer, some-layers, and all-layers models.
\end{theorem}

\begin{proof}
Our reduction is a modification of the assigned reduction of
Theorem~\ref{thm:assigned fix}.
In their paper, Akitaya et al.~\cite{simple} mention that there is only one place where the assignments are necessary: all of the creases on the Wrapper must be all valley (or all mountain).
The only other creases are those on the Staircase.
These creases are alternating in assignment, but their assignment is not relevant: if the construction is not foldable, changing the assignment does not help because all that the construction cares about is the length of the folded segments.

The uniform assignment of the Wrapper forces all of the creases on the Wrapper to be folded in order from right to left. This is necessary to ensure that the Bar passes through the Cage at each step, which checks to make sure that the 3-Partition instance was satisfied. In \cite{simple}, the authors show how to use unassigned $45^\circ$ creases to force the assignment of these creases, showing NP-hardness with unassigned $45^\circ$ creases. We will instead use unassigned orthogonal creases and particularly shaped paper to achieve this.

\begin{figure}
\includegraphics[width=\linewidth]{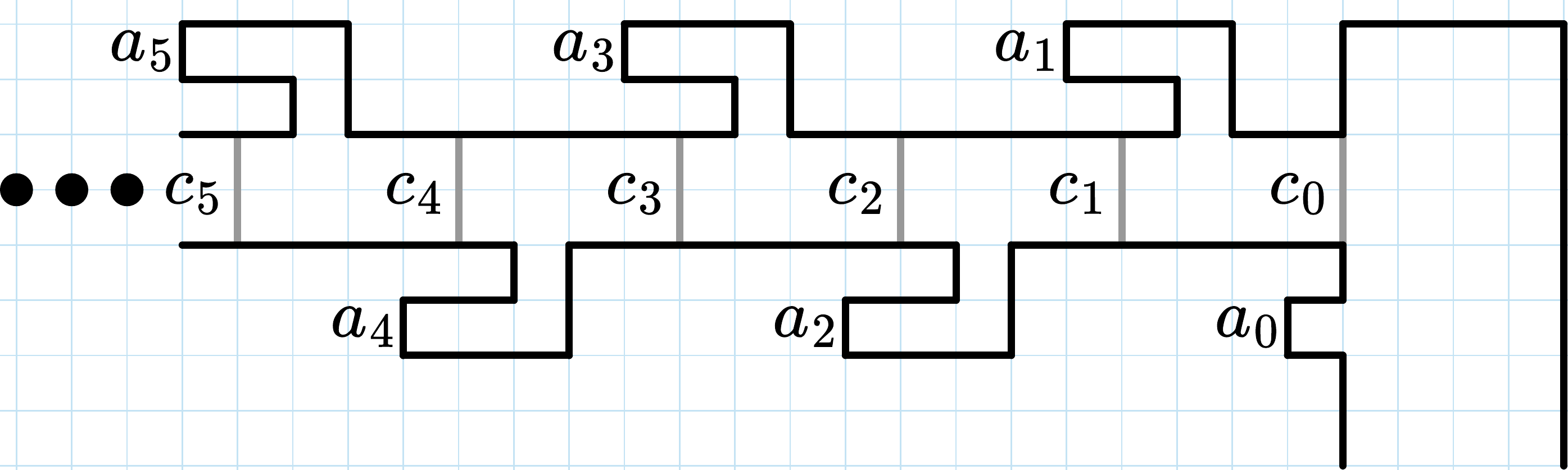}
\caption{This construction, which we call the \defn{cactus}, replaces the Wrapper. A right-to-left fold order is forced by the ``arms'' sticking out of the cactus. For example, if crease $c_1$ is folded before crease $c_0$, it will cause $a_2$ and $a_0$ to lie on top of each other, and regardless of the assignments $c_0$ will never be able to fold without $a_2$ and $a_0$ colliding.}
\label{fig:cactus}
\end{figure}

We replace the Wrapper with Figure~\ref{fig:cactus}, which we call the Cactus. On the Cactus, just before each crease $c_i$ we include a short branch $a_i$ of paper sticking out which the crosses the line containing the $c_i$. The point of these branches is to ensure that the creases are folded in increasing order, beginning with $c_0$. If any paper which is connected to the opposite side of $c_i$ is folded on top of $a_i$, then $c_i$ cannot be folded since folding it in either direction would force the paper to intersect itself. The Cactus's branches are aligned so that whenever a crease $c_i$ is folded, the previous branch $a_{i-1}$ will be aligned with the next branch $a_{i+1}$. This makes it impossible to fold $c_{i-1}$, so any valid folding must fold $c_{i-1}$ before $c_i$, and thus respect the intended order on the Cactus.

The purpose of Arm~1 is to enforce that all creases on the Wrapper are folded in the same direction, assuming the order of these creases is already forced. In particular, if the (unassigned) creases of the Wrapper are folded from right to left, the only way for Arm 1 to never collide with anything is for those creases to be either all mountain or all valley creases. This is needed to ensure that the Bar passes through the Cage on each step (see Figure~\ref{old-orthogonal-finite-steps-figure}). Akitaya et al.\ use this to prove NP-hardness for unassigned $45^\circ$ creases.

We use Arm 1 in essentially the same way. Because the creases in the Cactus are forced to be folded in the correct order, Arm 1 ensures that they are all folded with the same assignment.
Similarly, Arm 2 functions in the same way as in the previous reduction.
If the previous fold through a Cactus crease brings the Bar to overlap with the Cage, no fold through a Staircase crease is possible. 
Since we can't fold through another Cactus crease either, no folds are possible.

Finally, our construction is independent of which of the one-layer, some-layers, or all-layers model is used. Every fold in the construction is a one layer fold, so any instance that is foldable will be foldable in one-layer, some-layers, and all-layers models. 
\end{proof}

\section{Open Problems}

As mentioned in Section~\ref{sec:remaining-case}, when restricted to orthogonal
crease patterns and the parameters we consider, the sole remaining open problem
is infinite all-layers simple folds of assigned crease patterns on
arbitrary or orthogonal pieces of paper.
In this section, we present some observations and approaches to this problem.


One natural approach to proving hardness is to mimic the reduction from the
rectangular mixed-assignment problem solved in \cite{infinitealllayers}; see
Figure~\ref{old-mixed-figure}.
We can modify this reduction by removing a small neighborhood of paper
around each unassigned crease.
If the paper is still connected, then folding the resulting orthogonal
fully assigned crease pattern is equivalent to the original problem. 
Unfortunately, this process often disconnects the paper, and in particular
it does when applied to the 3SAT reduction of Figure~\ref{old-mixed-figure}.
This barrier seems difficult to overcome.

A related approach to a polynomial-time algorithm is to reduce to
the rectangular fully assigned problem, which is solvable in polynomial time.
We can naturally expand the piece of paper to its rectangular bounding box,
and extend all creases into the bounding box as unassigned creases,
as shown on the left of Figure~\ref{crease extension}.
To complete this reduction, we would need a way to assign all of these
extended creases in a way that makes the crease pattern foldable with infinite all-layers folds.
But it is not clear how to do this efficiently, and there may be complicated interactions between constraints on the assignment. Figure~\ref{orth hard examples} demonstrates that this is harder than the fully assigned case, since it is possible to make the paper not foldable by folding a crease which appears safe.

\begin{figure}
  \centering
  \includegraphics[scale=0.5]{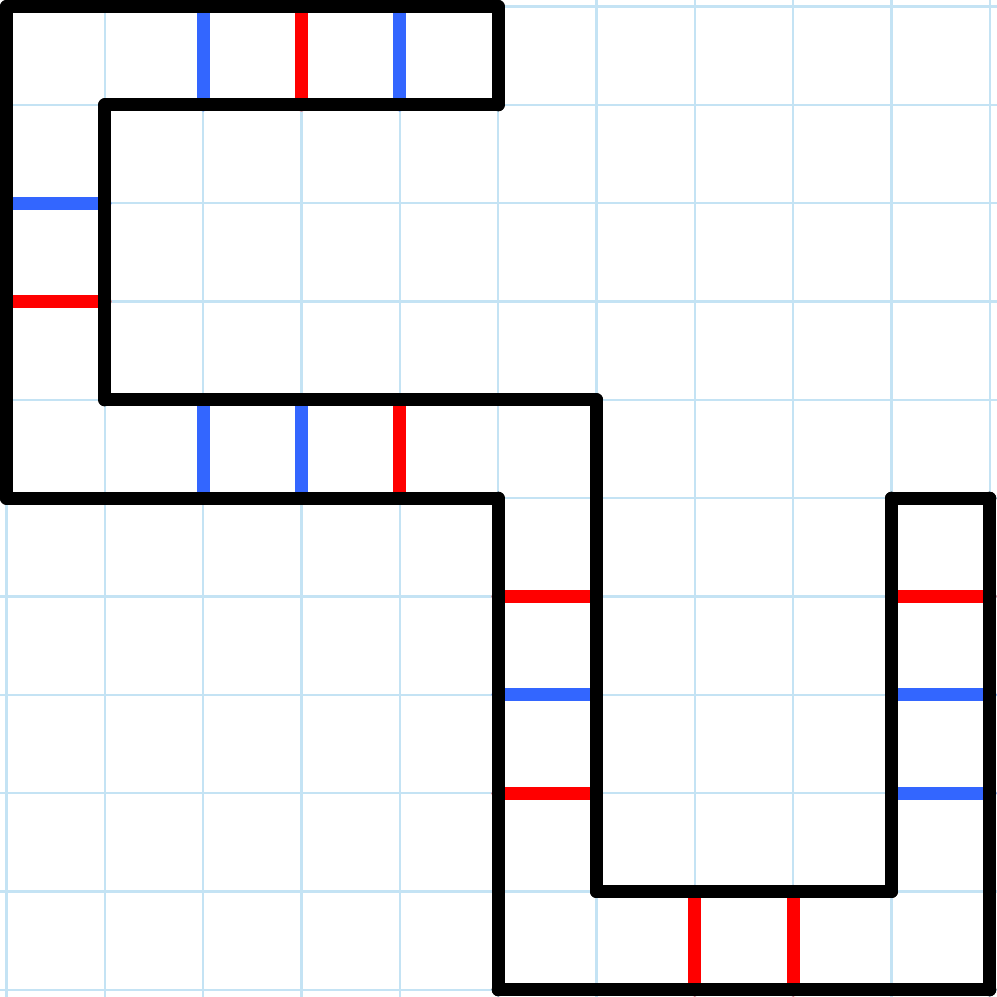}\hfil
  \includegraphics[scale=0.5]{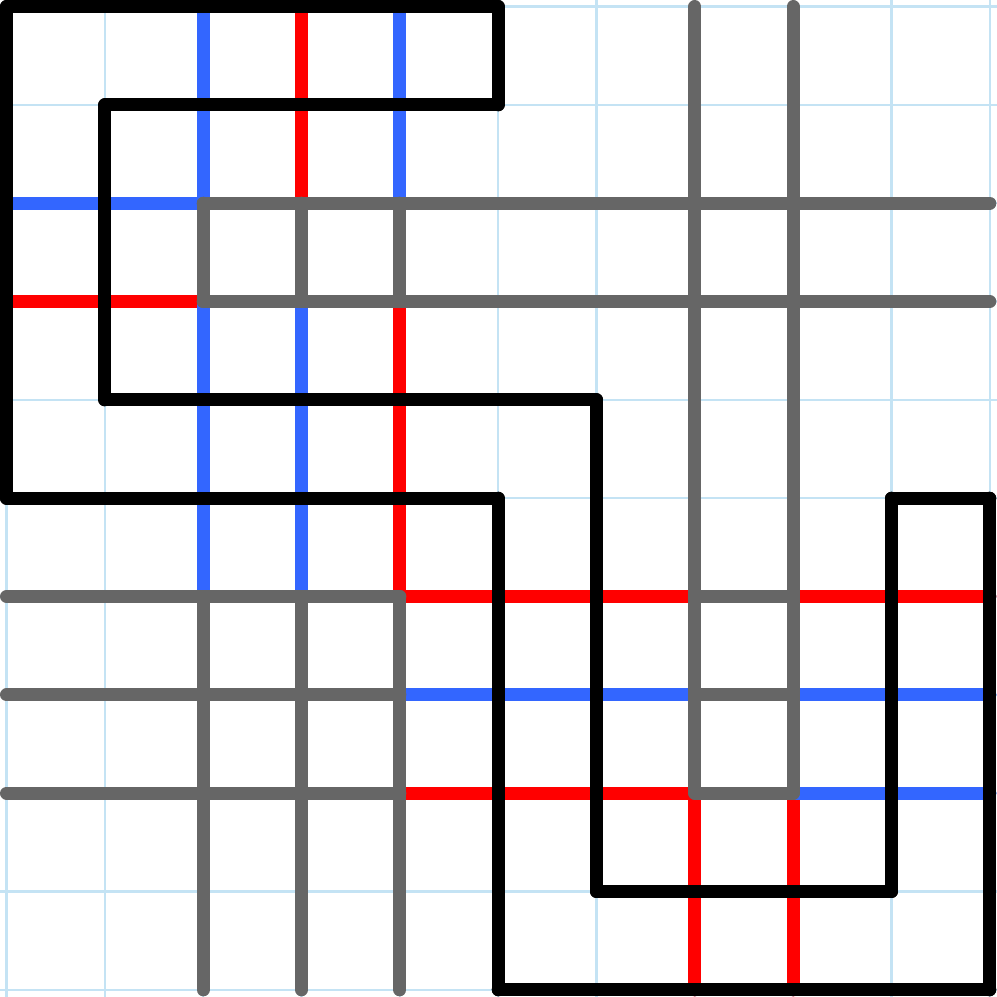}
  \caption{Extending creases into the rectangular bounding box. Each segment of
    a crease must have a uniform color until it meets a perpendicular crease.}
  \label{crease extension}
\end{figure}



\begin{figure}
  \centering
  \includegraphics[scale=0.5]{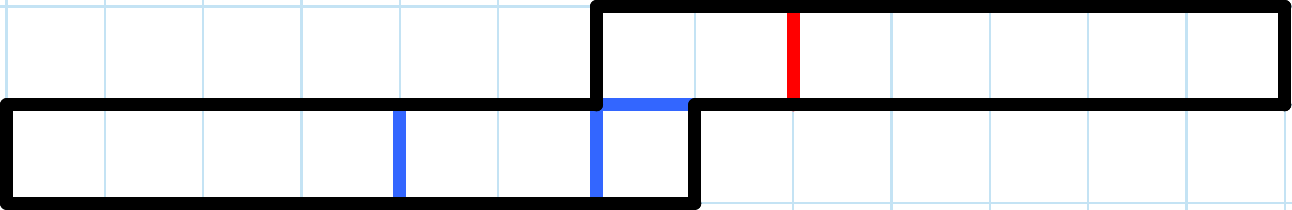}\hfil
  \includegraphics[scale=0.5]{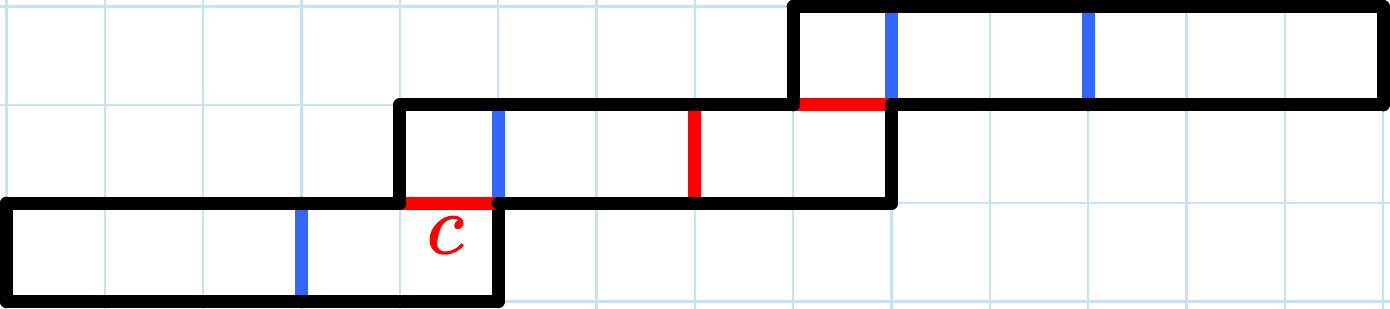}
  \caption{Left: A foldable pattern where folding the horizontal crease seems
    safe (does not cause immediate conflicts), but makes the pattern unfoldable.
    Right: A foldable pattern where crease $c$ appears safe even when
    considering 1D mixed foldability of every row and column,
    yet folding it breaks foldability.}
  \label{orth hard examples}
\end{figure}

\section*{Acknowledgments}

This work was initiated during an MIT class on Geometric Folding Algorithms
(6.849, Fall 2020).  We thank the other participants of that class ---
in particular, Walker Anderson, Joshua Ani, Lily Chung, Vincent Huang,
Jeffery Li, and Jamie Tucker-Foltz
--- for helpful discussions and providing a productive research environment.

\bibliographystyle{alpha}
\bibliography{citations.bib}

\end{document}